\documentclass{eptcs}
\providecommand{\event}{WORDS 2011} 

\usepackage{amsmath,amsthm,amssymb}
\usepackage[dvips]{epsfig}
\usepackage{epsf}
\usepackage{breakurl}
\usepackage{float}
\usepackage{array}
\usepackage{gastex}

\def\tbar{{\overline{\bf t}}}
\def\eff{{\cal F}_{7/3}}

\def\fmod#1 #2{#1\ ({\rm mod}\ #2)}
\def\gfg{$\frac{7}{3}$-power-free }
\def\gfe{$\frac{7}{3}$-power-free}

\theoremstyle{plain}
\newtheorem{theorem}{Theorem}
\newtheorem{corollary}[theorem]{Corollary}
\newtheorem{lemma}[theorem]{Lemma}

\theoremstyle{definition}

\theoremstyle{remark}
\newtheorem{remark}[theorem]{Remark}

\title{Fife's Theorem for $\frac{7}{3}$-Powers}

\author{Narad Rampersad
\institute{
Department of Mathematics,
University of Li{\`e}ge,
Grande Traverse, 12 (Bat.\ B37),
4000 L{i\`e}ge,
Belgium}
\email{narad.rampersad@gmail.com}
\and
Jeffrey Shallit
\institute{
School of Computer Science,
University of Waterloo,
Waterloo, ON  N2L 3G1, Canada}
\email{shallit@cs.uwaterloo.ca}
\and
Arseny Shur
\institute{Department of Algebra and Discrete Mathematics,
Ural State University, 
Ekaterinburg, Russia}
\email{arseny.shur@usu.ru}
}

\def\titlerunning{Fife's Theorem for $7/3$-Powers}
\def\authorrunning{N. Rampersad, J. Shallit, and A. Shur}

\begin{document}

\maketitle

\begin{abstract}
We prove a Fife-like characterization of the infinite binary
$\frac{7}{3}$-power-free words, by
giving a finite automaton of $15$ states that encodes all such words.
As a consequence, we characterize all such words that are
$2$-automatic.
\end{abstract}

\section{Introduction}

An {\it overlap} is a word of the form $axaxa$, where $a$ is a single
letter and $x$ is a (possibly empty) word.  In 1980, Earl Fife
\cite{Fife:1980} proved a theorem characterizing the infinite binary
overlap-free words as encodings of paths in a finite automaton.
Berstel \cite{Berstel:1994} later simplified the exposition, and both
Carpi \cite{Carpi:1993a} and Cassaigne \cite{Cassaigne:1993b} gave an
analogous analysis for the case of finite words.

In a previous paper \cite{Shallit:2011}, the second author gave a new
approach to Fife's theorem, based on the factorization theorem
of Restivo and Salemi \cite{Restivo&Salemi:1985a} for overlap-free words.
In this paper, we extend this analysis
by applying it to the case of \gfg words.

Given a rational number $\frac{p}{q} > 1$, we define a word $w$ to be a
$\frac{p}{q}$-power if $w$ can be written in the form $x^n x'$ where $n
= \lfloor p/q \rfloor$, $x'$ is a (possibly empty)
prefix of $x$, and $|w|/|x| = p/q$.
The word $x$ is called a {\it period} of $w$, and $p/q$ is an 
{\it exponent} of $w$.  If $p/q$ is the largest exponent of $w$,
we write $\exp(w) = p/q$.   We also say that $w$ is {\it $|x|$-periodic}.
For example, the word
{\tt alfalfa} is a $\frac{7}{3}$-power, and the corresponding period is
{\tt alf}.  Sometimes, as is routine in the literature,
we also refer to $|x|$ as the period; the context
should make it clear which is meant.

A word, whether finite or infinite, is {\it $\beta$-power-free}
if it contains no factor $w$ that is an $\alpha$-power for $\alpha\geq
\beta$.    A word is {\it $\beta^+$-power-free} if it contains no
factor $w$ that is an $\alpha$-power for $\alpha > \beta$.  Thus, the
concepts of ``overlap-free'' and ``$2^+$-power-free'' coincide.

\section{Notation and basic results}

Let $\Sigma$ be a finite alphabet.  We let $\Sigma^*$ denote the set
of all finite words over $\Sigma$ and $\Sigma^\omega$ denote the
set of all (right-) infinite words over $\Sigma$.  We say
$y$ is a {\em factor} of a word $w$ if there exist words
$x, z$ such that $w = xyz$.

If $x$ is a finite word, then $x^\omega$ represents the
infinite word $xxx \cdots$.

From now on we fix $\Sigma = \lbrace 0,1 \rbrace$.  The most famous
infinite binary overlap-free word is $\bf t$, the Thue-Morse word,
defined as the fixed point, starting with $0$, of the Thue-Morse
morphism $\mu$, which maps $0$ to $01$ and $1$ to $10$.  We have
$$ {\bf t} = t_0 t_1 t_2 \cdots = 0110100110010110 \cdots .$$
The morphism $\mu$ has a second fixed point, $\tbar$, which is
obtained from $\bf t$ by applying the complementation coding
defined by $\overline{0} = 1$ and $\overline{1} = 0$.

We let $\eff$ denote the set of (right-) infinite binary 
\gfg words.  We point out that these words are of
particular interest, because $\frac{7}{3}$ is the largest exponent $\alpha$
such that there are only polynomially-many $\alpha$-power-free words of
length $n$ \cite{Karhumaki&Shallit:2004}.
The exponent $\frac{7}{3}$ plays a special role in 
combinatorics on words, as testified to by the many papers mentioning
this exponent (e.g.,
\cite{Kolpakov&Kucherov:1997,Shur:2000,Karhumaki&Shallit:2004,Rampersad:2005,Aberkane&Currie:2005,Blondel&Cassaigne&Jungers:2009}).  

We now state a
factorization theorem for infinite \gfg words:

\begin{theorem}
Let ${\bf x} \in \eff$, and let $P = \lbrace p_0, p_1, p_2, p_3, p_4 
\rbrace$, where $p_0 = \epsilon$, $p_1 = 0$, $p_2 = 00$,
$p_3 = 1$, and $p_4 = 11$.  Then there exists ${\bf y} \in \eff$ and
$p \in P$ such that ${\bf x} = p \mu({\bf y})$.  Furthermore, this
factorization is unique, and $p$ is uniquely determined by inspecting
the first $5$ letters of ${\bf x}$.
\end{theorem}

\begin{proof}
The first two claims follow immediately from the version for finite
words, as given in \cite{Karhumaki&Shallit:2004}.  The last claim
follows from exhaustive enumeration of cases.  
\end{proof}

We can now iterate this factorization theorem to get

\begin{corollary}
Every infinite \gfg word $\bf x$ can be written 
uniquely in the form
\begin{equation}
 {\bf x} = p_{i_1} \mu( p_{i_2} \mu ( p_{i_3} \mu ( \cdots ) ) ) \label{qq}
\end{equation}
with $i_j \in \lbrace 0, 1, 2, 3, 4 \rbrace$ for 
$j \geq 1$, subject to the understanding
that if there exists $c$ such that $i_j = 0$ for
$j \geq c$, then we also need to specify whether the ``tail'' of the
expansion represents $\mu^\omega(0) = {\bf t}$ or $\mu^\omega(1) = \tbar$.
Furthermore, every truncated expansion
$$p_{i_1} \mu(p_{i_2} \mu (p_{i_3} \mu (\cdots p_{i_{n-1}} \mu(p_{i_n})
\cdots )))$$
is a prefix of $\bf x$, with the understanding that if
$i_n = 0 $, then we need to replace $p_{i_n}$ with either
$1$ (if the ``tail'' represents $\bf t$) or $3$ (if the ``tail''
represents $\tbar$).
\end{corollary}

\begin{proof}
The form (\ref{qq}) is unique, since each $p_i$ is uniquely determined
by the first 5 characters of the associated word.
\end{proof}

Thus, we can associate each infinite binary \gfg
word $\bf x$ with the
essentially unique infinite sequence
of indices ${\bf i} := (i_j)_{j \geq 0}$ coding elements in $P$,
as specified by (\ref{qq}).  If $\bf i$ ends in $0^\omega$, then
we need an additional element (either $1$ or $3$) to disambiguate
between $\bf t$ and $\tbar$ as the ``tail''.  
In our notation, we separate this additional element with a
semicolon so that, for example, the string $000\cdots; 1$ represents
$\bf t$ and $000\cdots; 3$ represents $\tbar$.

Of course, not every possible sequence of $(i_j)_{j \geq 1}$ of indices
corresponds to an
infinite \gfg word.  For example, every infinite word coded
by an infinite sequence beginning $400\cdots$ has a $\frac{7}{3}$-power.
Our goal is to characterize precisely,
using a finite automaton, those infinite sequences corresponding to
\gfg words.

Next, we recall some connections between the morphism $\mu$ and the
powers over the binary alphabet. Below $x$ is an arbitrary finite or
right-infinite word.

\begin{lemma}
\label{squares}
If the word $\mu(x)$ has a prefix $zz$, then the word $x$ has the prefix $\mu^{-1}(z)\mu^{-1}(z)$.
\end{lemma}

\begin{proof}
Follows immediately from \cite[Lemma 1.7.2]{Allouche&Shallit:2003}.
\end{proof}

\begin{lemma} 
\label{exps}
(1) For any real $\beta>1$, we have $\exp(x)=\beta$ iff $\exp(\mu(x))=\beta$.\\
(2) For any real $\beta\ge2^+$, the word $x$ is $\beta$-power-free iff $\mu(x)$ is $\beta$-power-free.
\end{lemma}

\begin{proof}
For (1), see \cite[Prop.~1.1]{Shur:2000}. For (2), see \cite[Prop.~1.2]{Shur:2000} or \cite[Thm.~5]{Karhumaki&Shallit:2004}.
\end{proof}

\begin{lemma} \label{cor}
Let $p$ be a positive integer. If the longest $p$-periodic prefix of
the word $\mu(x)$ has the exponent $\beta \geq 2$, then the longest
$(p/2)$-periodic prefix of $x$ also has the exponent $\beta$.
\end{lemma}

\begin{proof}
Let $zzz'$ (where $|z|=p$ and $z'$ is a possibly empty prefix of
$z^{\omega}$) be the longest $p$-periodic prefix of $\mu(x)$.
Lemma~\ref{squares} implies that $p$ is even. If $|z'|$ is odd, let $a$
be the last letter of $z'$. The next letter $b$ in $\mu(x)$ is fixed by
the definition of $\mu$: $b\ne a$. By the definition of period, another $a$
occurs $p$ symbols to the left of the last letter of $z'$. Since $p$ is
even, this $a$ also fixes the next letter $b$. Hence the prefix
$zzz'b$ of $\mu(x)$ is $p$-periodic, contradicting the
definition of $zzz'$. Thus $|z'|$ is even. Therefore $x$ begins with
the $\beta$-power $\mu^{-1}(z)\mu^{-1}(z)\mu^{-1}(z')$ of period
$p/2$.

It remains to note that if $x$ has a $(p/2)$-periodic prefix $y$ of
exponent $\alpha>\beta$, then by Lemma~\ref{exps}\,(1), the
$p$-periodic prefix $\mu(y)$ of $\mu(x)$ also has the exponent $\alpha$,
contradicting the hypotheses of the lemma.  
\end{proof}

\section{The main result}

For each finite word $w \in \lbrace 0,1,2,3,4 \rbrace^*$,
$w = i_1 i_2 \cdots i_r$, and an infinite word
${\bf x} \in \lbrace 0, 1 \rbrace^\omega$, we
define 
\begin{align*}
C_w ({\bf x}) &= p_{i_1} \mu( p_{i_2} \mu ( p_{i_3} \mu( \cdots {\bf x} \cdots))) \text{ and}\\
F_w &= \lbrace {\bf x} \in \Sigma^\omega \ : \  C_w ({\bf x}) \in \eff \rbrace .
\end{align*}

Note that $F_w\subseteq \eff$ for any $w$ in view of Lemma~\ref{exps}\,(2).

\begin{figure}[htp]
\centerline{ 
\unitlength=1mm
\begin{picture}(100,180)(-1,-90)
\gasset{Nw=6,Nh=4.5,AHnb=0,Nframe=n}
\node(eps)(0,0){\small$F_{\epsilon}$} 
\node(0)(12,0){\small$F_0$}
\node(0a)(18.5,0){\small$=F_{\epsilon}$}
\node(1)(12,61){\small$F_1$}
\node(2)(12,16){\small$F_2$}
\drawedge(eps,0){}
\drawedge(eps,1){}
\drawedge(eps,2){}
\node(10)(30,73){\small$F_{10}$} 
\node(10a)(36.5,73){\small$=F_{\epsilon}$}
\node(11)(30,67){\small$F_{11}$}
\node(12)(30,61){\small$F_{12}$} 
\node(12a)(36.5,61){\small$=F_2$}
\node(13)(30,55){\small$F_{13}$}
\node(14)(30,49){\small$F_{14}$} 
\node(14a)(36.5,49){\small$=\varnothing$}
\node(21)(30,28){\small$F_{21}$} 
\node(21a)(36.5,28){\small$=\varnothing$} 
\node(22)(30,22){\small$F_{22}$}
\node(22a)(36.5,22){\small$=\varnothing$}
\node(20)(30,16){\small$F_{20}$} 
\node(23)(30,10){\small$F_{23}$}
\node(23a)(37,10){\small$=F_{13}$}
\node(24)(30,4){\small$F_{24}$} 
\node(24a)(36.5,4){\small$=\varnothing$}
\drawedge(1,10){}
\drawedge(1,11){}
\drawedge(1,12){}
\drawedge(1,13){}
\drawedge(1,14){}
\drawedge(2,20){}
\drawedge(2,21){}
\drawedge(2,22){}
\drawedge(2,23){}
\drawedge(2,24){}
\node[Nw=7](110)(55,88){\small$F_{110}$} 
\node(110a)(62,88){\small$=F_3$}
\node[Nw=7](111)(55,82){\small$F_{111}$} 
\node(111a)(62.5,82){\small$=F_{11}$}
\node[Nw=7](112)(55,76){\small$F_{112}$} 
\node(112a)(62,76){\small$=F_2$}
\node[Nw=7](113)(55,70){\small$F_{113}$} 
\node(113a)(62.5,70){\small$=F_{13}$}
\node[Nw=7](114)(55,64){\small$F_{114}$} 
\node(114a)(62,64){\small$=\varnothing$}
\node[Nw=7](130)(55,58){\small$F_{130}$} 
\node[Nw=7](131)(55,52){\small$F_{131}$} 
\node(131a)(62.5,52){\small$=F_{31}$}
\node[Nw=7](132)(55,46){\small$F_{132}$} 
\node(132a)(62,46){\small$=\varnothing$}
\node[Nw=7](133)(55,40){\small$F_{133}$} 
\node(133a)(62,40){\small$=\varnothing$}
\node[Nw=7](134)(55,34){\small$F_{134}$} 
\node(134a)(62,34){\small$=\varnothing$}
\node[Nw=7](200)(55,28){\small$F_{200}$} 
\node(200a)(62,28){\small$=\varnothing$}
\node[Nw=7](201)(55,22){\small$F_{201}$} 
\node(201a)(62,22){\small$=\varnothing$}
\node[Nw=7](202)(55,16){\small$F_{202}$} 
\node(202a)(62,16){\small$=\varnothing$}
\node[Nw=7](203)(55,10){\small$F_{203}$} 
\node[Nw=7](204)(55,4){\small$F_{204}$} 
\node(204a)(62,4){\small$=F_{4}$}
\drawedge[curvedepth=-3.5](11,110){}
\drawedge[curvedepth=-2.7](11,111){}
\drawedge[curvedepth=-2](11,112){}
\drawedge[curvedepth=-1](11,113){}
\drawedge(11,114){}
\drawedge(13,130){}
\drawedge[curvedepth=1](13,131){}
\drawedge[curvedepth=2](13,132){}
\drawedge[curvedepth=2.7](13,133){}
\drawedge[curvedepth=3.5](13,134){}
\drawedge[curvedepth=-1.5](20,200){}
\drawedge[curvedepth=-0.8](20,201){}
\drawedge(20,202){}
\drawedge[curvedepth=0.8](20,203){}
\drawedge[curvedepth=1.5](20,204){}
\node[Nw=8](1300)(80,73){\small$F_{1300}$} 
\node(1300a)(89,73){\small$=F_{130}$}
\node[Nw=8](1301)(80,67){\small$F_{1301}$}
\node(1301a)(88,67){\small$=F_1$}
\node[Nw=8](1302)(80,61){\small$F_{1302}$} 
\node(1302a)(88,61){\small$=F_2$}
\node[Nw=8](1303)(80,55){\small$F_{1303}$}
\node(1303a)(87.5,55){\small$=\varnothing$}
\node[Nw=8](1304)(80,49){\small$F_{1304}$} 
\node(1304a)(87.5,49){\small$=\varnothing$}
\node[Nw=8](2030)(80,28){\small$F_{2030}$} 
\node(2030a)(89,28){\small$=F_{310}$} 
\node[Nw=8](2031)(80,22){\small$F_{2031}$}
\node(2031a)(87.5,22){\small$=\varnothing$}
\node[Nw=8](2032)(80,16){\small$F_{2032}$} 
\node(2032a)(87.5,16){\small$=\varnothing$}
\node[Nw=8](2033)(80,10){\small$F_{2033}$} 
\node(2033a)(88.5,10){\small$=F_{33}$}
\node[Nw=8](2034)(80,4){\small$F_{2034}$} 
\node(2034a)(88,4){\small$=F_4$}
\drawedge[curvedepth=-1.5](130,1300){}
\drawedge[curvedepth=-0.8](130,1301){}
\drawedge(130,1302){}
\drawedge[curvedepth=0.8](130,1303){}
\drawedge[curvedepth=1.5](130,1304){}
\drawedge[curvedepth=-2](203,2030){}
\drawedge[curvedepth=-1.2](203,2031){}
\drawedge[curvedepth=-0.6](203,2032){}
\drawedge(203,2033){}
\drawedge[curvedepth=0.8](203,2034){}
\node(3)(12,-61){\small$F_3$}
\node(4)(12,-16){\small$F_4$}
\drawedge(eps,3){}
\drawedge(eps,4){}
\node(30)(30,-73){\small$F_{30}$} 
\node(30a)(36.5,-73){\small$=F_{\epsilon}$}
\node(33)(30,-67){\small$F_{33}$}
\node(34)(30,-61){\small$F_{34}$} 
\node(34a)(36.5,-61){\small$=F_4$}
\node(31)(30,-55){\small$F_{31}$}
\node(32)(30,-49){\small$F_{32}$} 
\node(32a)(36.5,-49){\small$=\varnothing$}
\node(43)(30,-28){\small$F_{43}$} 
\node(43a)(36.5,-28){\small$=\varnothing$} 
\node(44)(30,-22){\small$F_{44}$}
\node(44a)(36.5,-22){\small$=\varnothing$}
\node(40)(30,-16){\small$F_{40}$} 
\node(41)(30,-10){\small$F_{41}$}
\node(41a)(37,-10){\small$=F_{31}$}
\node(42)(30,-4){\small$F_{42}$} 
\node(42a)(36.5,-4){\small$=\varnothing$}
\drawedge(3,30){}
\drawedge(3,33){}
\drawedge(3,34){}
\drawedge(3,31){}
\drawedge(3,32){}
\drawedge(4,40){}
\drawedge(4,43){}
\drawedge(4,44){}
\drawedge(4,41){}
\drawedge(4,42){}
\node[Nw=7](330)(55,-88){\small$F_{330}$} 
\node(330a)(62,-88){\small$=F_1$}
\node[Nw=7](333)(55,-82){\small$F_{333}$} 
\node(333a)(62.5,-82){\small$=F_{33}$}
\node[Nw=7](334)(55,-76){\small$F_{334}$} 
\node(112a)(62,-76){\small$=F_4$}
\node[Nw=7](331)(55,-70){\small$F_{331}$} 
\node(331a)(62.5,-70){\small$=F_{31}$}
\node[Nw=7](332)(55,-64){\small$F_{332}$} 
\node(332a)(62,-64){\small$=\varnothing$}
\node[Nw=7](310)(55,-58){\small$F_{310}$} 
\node[Nw=7](313)(55,-52){\small$F_{313}$} 
\node(313a)(62.5,-52){\small$=F_{13}$}
\node[Nw=7](314)(55,-46){\small$F_{314}$} 
\node(314a)(62,-46){\small$=\varnothing$}
\node[Nw=7](311)(55,-40){\small$F_{311}$} 
\node(311a)(62,-40){\small$=\varnothing$}
\node[Nw=7](312)(55,-34){\small$F_{312}$} 
\node(312a)(62,-34){\small$=\varnothing$}
\node[Nw=7](400)(55,-28){\small$F_{400}$} 
\node(400a)(62,-28){\small$=\varnothing$}
\node[Nw=7](403)(55,-22){\small$F_{403}$} 
\node(403a)(62,-22){\small$=\varnothing$}
\node[Nw=7](404)(55,-16){\small$F_{404}$} 
\node(404a)(62,-16){\small$=\varnothing$}
\node[Nw=7](401)(55,-10){\small$F_{401}$} 
\node[Nw=7](402)(55,-4){\small$F_{402}$} 
\node(402a)(62,-4){\small$=F_{2}$}
\drawedge[curvedepth=3.5](33,330){}
\drawedge[curvedepth=2.7](33,333){}
\drawedge[curvedepth=2](33,334){}
\drawedge[curvedepth=1](33,331){}
\drawedge(33,332){}
\drawedge(31,310){}
\drawedge[curvedepth=-1](31,313){}
\drawedge[curvedepth=-2](31,314){}
\drawedge[curvedepth=-2.7](31,311){}
\drawedge[curvedepth=-3.5](31,312){}
\drawedge[curvedepth=1.5](40,400){}
\drawedge[curvedepth=0.8](40,403){}
\drawedge(40,404){}
\drawedge[curvedepth=-0.8](40,401){}
\drawedge[curvedepth=-1.5](40,402){}
\node[Nw=8](3100)(80,-73){\small$F_{3100}$} 
\node(3100a)(89,-73){\small$=F_{310}$}
\node[Nw=8](3103)(80,-67){\small$F_{3103}$}
\node(3103a)(88,-67){\small$=F_3$}
\node[Nw=8](3104)(80,-61){\small$F_{3104}$} 
\node(3104a)(88,-61){\small$=F_4$}
\node[Nw=8](3101)(80,-55){\small$F_{3101}$}
\node(3101a)(87.5,-55){\small$=\varnothing$}
\node[Nw=8](3102)(80,-49){\small$F_{3102}$} 
\node(3102a)(87.5,-49){\small$=\varnothing$}
\node[Nw=8](4010)(80,-28){\small$F_{4010}$} 
\node(4010a)(89,-28){\small$=F_{130}$} 
\node[Nw=8](4013)(80,-22){\small$F_{4013}$}
\node(4013a)(87.5,-22){\small$=\varnothing$}
\node[Nw=8](4014)(80,-16){\small$F_{4014}$} 
\node(4014a)(87.5,-16){\small$=\varnothing$}
\node[Nw=8](4011)(80,-10){\small$F_{4011}$} 
\node(4011a)(88.5,-10){\small$=F_{11}$}
\node[Nw=8](4012)(80,-4){\small$F_{4012}$} 
\node(4012a)(88,-4){\small$=F_2$}
\drawedge[curvedepth=1.5](310,3100){}
\drawedge[curvedepth=0.8](310,3103){}
\drawedge(310,3104){}
\drawedge[curvedepth=-0.8](310,3101){}
\drawedge[curvedepth=-1.5](310,3102){}
\drawedge[curvedepth=2](401,4010){}
\drawedge[curvedepth=1.2](401,4013){}
\drawedge[curvedepth=0.6](401,4014){}
\drawedge(401,4011){}
\drawedge[curvedepth=-0.8](401,4012){}
\end{picture} }
\caption{\small\sl Equations between languages $F_w$.} \label{tree}
\end{figure}
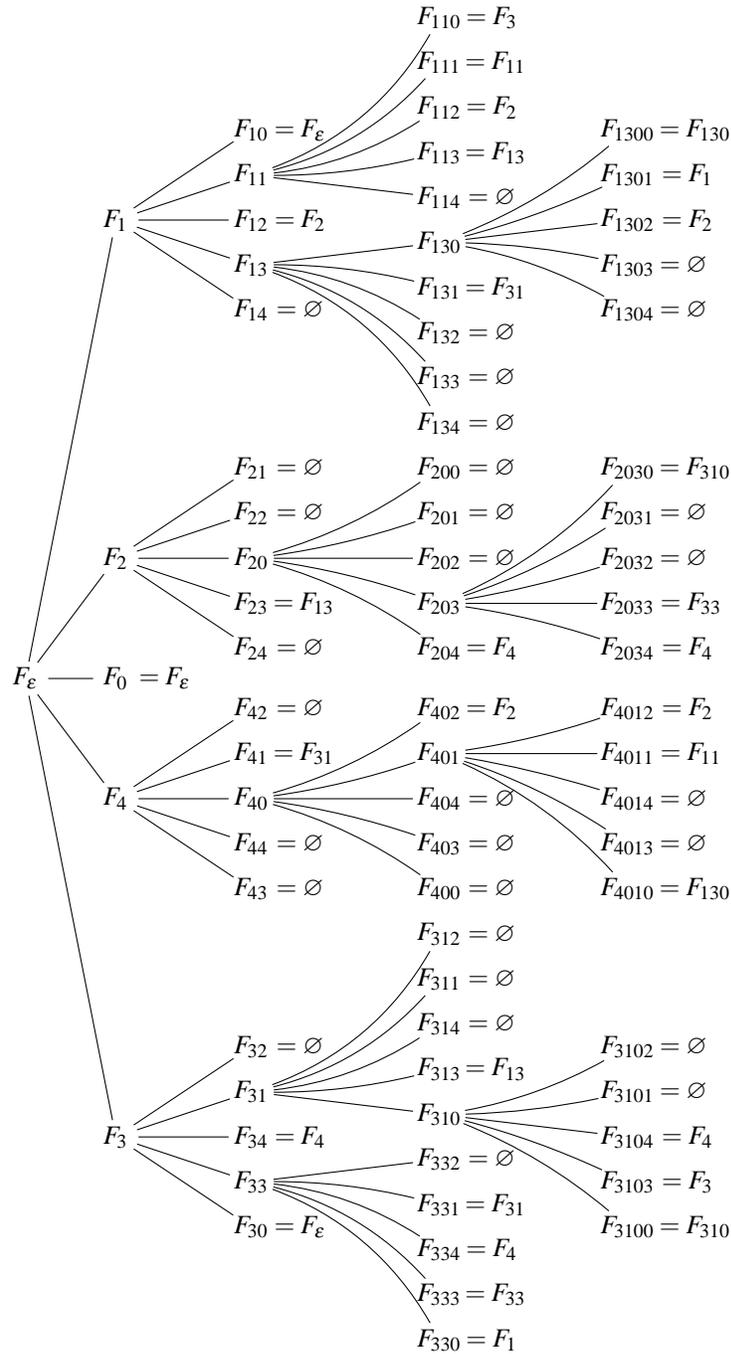

\begin{lemma} \label{main2}
The sets $F_w$ satisfy the equalities listed in Fig.~\ref{tree}. In particular, there are only 15 different nonempty sets $F_w$; they are
$$
F_{\epsilon},F_1,F_{11},F_{13},F_{130},F_2,F_{20},F_{203},F_3,F_{31},F_{310},F_{33},F_4,F_{40},F_{401}.
$$
\end{lemma}

\begin{proof}
Due to symmetry, it is enough to prove only the 30 equalities from the
upper half of Fig.~\ref{tree} and the equality $F_0=F_{\epsilon}$. We
first prove the emptiness of 15 sets from the upper half of
Fig.~\ref{tree}.

\smallskip
Four sets: $F_{21}$, $F_{22}$, $F_{201}$, and $F_{202}$, consist of
words that start $000$.

Eight sets consist of words that contain the factor $0\mu(11)=01010$
($F_{14}$, $F_{24}$, $F_{133}$, $F_{134}$, $F_{1303}$, and $F_{1304}$),
its $\mu$-image ($F_{114}$), or the complement of its $\mu$-image
($F_{132}$).

Two sets: $F_{2031}$ and $F_{2032}$, consist of words that start
$00\mu^2(1)0=0010010$. Finally, the words from the set $F_{200}$ have
the form $00\mu^3({\bf x})$; each of these words starts either $000$ or
$0010010$.

\medskip
Each of the 16 remaining equalities has the form $F_{w_1}=F_{w_2}$. We
prove them by showing that for an arbitrary ${\bf x}\in \eff$, the
words $u_1=C_{w_1}({\bf x})$ and $u_2=C_{w_2}({\bf x})$ are either both
\gfg or both not. In most cases, some suffix of $u_1$ coincides with
the image of $u_2$ under some power of $\mu$. Then by
Lemma~\ref{exps}\,(2) the word $u_1$ can be \gfg only if $u_2$ is \gfe.
In these cases, it suffices to study $u_1$ assuming that $u_2$ is
\gfe.

When we refer to a ``forbidden'' power in what follows, we mean 
a power of exponent $\geq \frac{7}{3}$.

\smallskip

$F_0=F_{\epsilon}$: By Lemma~\ref{exps}\,(2), $u_1=\mu({\bf x})$ is
\gfg iff $u_2={\bf x}$ is \gfe.

\smallskip

$F_{10}=F_{\epsilon}$: The word $u_1=0\mu(\mu({\bf x}))$ contains a
$\mu^2$-image of $u_2={\bf x}$. If ${\bf x}$ is \gfe, then so is
$\mu^2({\bf x})$. Hence, if $u_1$ has a forbidden power, then this
power must be a prefix of $u_1$.

Now let $\beta<7/3$ be the largest possible exponent of a prefix of
${\bf x}$ and $q$ be the smallest period of a prefix of exponent
$\beta$ in ${\bf x}$. Write $\beta=p/q$. Then the word $\mu^2({\bf x})$
has a prefix of exponent $\beta$ and of period $4q$ by
Lemma~\ref{exps}\,(1), but no prefixes of a bigger exponent or of the
same exponent and a smaller period by Lemma~\ref{cor}. Hence $u_1$ has
no prefixes of exponent greater than $(4p+1)/(4q)$. Since $p$ and $q$
are integers, we obtain the required inequality:  $$
\frac{p}{q}<\frac{7}{3}\Longrightarrow3p<7q\Longrightarrow
3p+\frac{3}{4}<7q\Longrightarrow \frac{4p+1}{4q}<\frac{7}{3}.  $$

\smallskip

$F_{12}=F_2$: The word $u_1=0\mu(00\mu({\bf x}))$ contains a
$\mu$-image of $u_2=00\mu({\bf x})$. Suppose that $u_2$ is \gfe. Then
it starts $0010011$. Since the factor $001001$ cannot occur in a
$\mu$-image, we note that

\begin{itemize}
\item[$(\star)$] the word $00\mu({\bf x})$ has only two prefixes of exponent 2 ($00$ and $001001$) and no prefixes of bigger exponents.
\end{itemize}

By Lemma~\ref{cor}, the word $\mu(u_2)$ has only two prefixes of
exponent 2 ($\mu(00)$ and $\mu(001001)$) and no prefixes of bigger
exponents. Thus, the word $u_1=0\mu(u_2)$ is obviously \gfe.

\smallskip

$F_{23}=F_{13}$: We have $u_1=00\mu(1\mu({\bf x}))=0u_2$. Suppose the
word $u_2$ is \gfe; then it starts $010011$. A forbidden power in
$u_1$, if any, occurs at the beginning and hence contains $0010011$.
But $00100$ does not occur later in this word, so no such forbidden
power exists.

\smallskip

$F_{110}=F_{3}$: The word $u_1=0\mu(0\mu(\mu({\bf x})))=001\mu^3(x)$ is
a suffix of the $\mu^2$-image of $u_2=1\mu({\bf x})$. Hence, if $u_2$
is \gfe, then by Lemma~\ref{exps}\,(2) $u_1$ is \gfg as well.

For the other direction, assume $u_1$ is \gfg and then $\mu({\bf x})$
is \gfe. So, if $u_2$ contains some power $y y y'$ with $|y'| \geq
|y|/3$, then this power must be a prefix of $u_2$. Put $y = 1z$ and
$y'=1z'$. The word $\mu({\bf x})$ starts $z1z1z'$. Hence ${\bf x}$
starts $\mu^{-1}(z1)\mu^{-1}(z1)$ by Lemma~\ref{squares}. So we
conclude that $|z1|=|y|$ is an even number. Now let $|y| = q$ and $p =
|yyy'|$ so that $p/q \geq 7/3$.  Thus the word $1u_1=\mu^2(u_2)$ starts
with a $(p/q)$-power of period $4q$. Since $u_1$ is \gfe, we have
$(4p{-}1)/4q < 7/3$.

This gives us the inequalities  $3p \geq 7q$ and $3p - 3/4 < 7q$.
Since $p$ and $q$ are integers this means $3p = 7q$ and hence $q$ is
divisible by 3.  On the other hand from above $q$ is even.  So $q$ is
divisible by 6. Now $|y'| = |y|/3$ so $|y'|$ is even. But then $z'$ is
odd, and begins at an even position in a $\mu$-image, so the character
following $z'$ is fixed and must be the same character as in the
corresponding position of $z$, say $a$.  Thus $z1z1z'a$ is a
$(7/3)$-power occurring in $\mu({\bf x})$, a contradiction.

\smallskip

$F_{111}=F_{11}$: The word $u_1=0\mu(0\mu(0\mu({\bf x})))=0010110 \mu^3
({\bf x})$ contains a $\mu$-image of $u_2=0\mu(0\mu({\bf
  x}))=001 \mu^2({\bf x})$. Suppose $u_2$ is \gfg but to the contrary
$u_1=0\mu(u_2)$ has a forbidden power. By Lemma~\ref{exps}\,(2), this
power must be a prefix of $u_1$. Note that this power can be extended
to the left by 1 (not by 0, because a $\mu$-image cannot contain
$000$). Hence the word $1u_1=\mu^3(1{\bf x})$ starts with a forbidden
power. This induces a forbidden power at the beginning of $1{\bf x}$;
this power has a period $q$ and some exponent $p/q\ge 7/3$. Then $u_1$
has a prefix of exponent $(8p-1)/8q\ge 7/3$. On the other hand the word
$\mu(u_2)$ is \gfe, whence $(8p-2)/q< 7/3$. So we get the system of
inequalities $3p - 3/8 \geq 7q$, $3p - 3/4 < 7q$. This system has no
integer solutions, a contradiction.

\smallskip

$F_{112}=F_{2}$: We have $u_1=0\mu(0\mu(00\mu({\bf
x})))=001\mu^2(00\mu({\bf x}))=001\mu^2(u_2)$. In view of $(\star)$,
one can easily check that if $u_2$ is \gfe, then so is $u_1$.

\smallskip

$F_{113}=F_{13}$: We have $u_1=0\mu(0\mu(1\mu({\bf
  x})))=0\mu(u_2)$. Suppose $u_2$ is \gfe. Then $\mu^2({\bf x})$ starts
$01101001$. Assume to the contrary that $u_1$ has a forbidden power. By
Lemma~\ref{exps}\,(2), this power must be a prefix of $u_1$. Again,
this power can be extended to the left by 1, not by 0. Hence the word
$1u_1=\mu^2(11\mu({\bf x}))$ starts with a forbidden power, thus
inducing a forbidden power at the beginning of $u=11\mu({\bf
x})=110110\cdots$. The word $u$ has only two squares as prefixes ($11$
and $110110$, cf. $(\star)$). Hence $u$ has the prefix $11011010$ and
no forbidden factors except for the $(7/3)$-power prefix $1101101$.
Therefore, the word $u_1$ has no prefixes of exponent $\ge 7/3$.

\smallskip

$F_{131}=F_{31}$: We have $u_1=0\mu(1\mu(0\mu({\bf x})))=0\mu(u_2)$.
Suppose $u_2$ is \gfe. Then the word $u=11\mu(0\mu({\bf x}))$ is \gfg
by the equality $F_{41}=F_{31}$, which is symmetric to $F_{23}=F_{13}$
proved above. But $u_1$ is a suffix of $\mu(u)$, whence the result.

\smallskip

$F_{204}=F_{4}$: We have $u_1=00\mu(\mu(11\mu({\bf x})))=00\mu^2(u_2)$.
Suppose $u_2$ is \gfe. Using the observation symmetric to $(\star)$, we
check by inspection that $u_1$ contains no forbidden power.

\smallskip

$F_{1300}=F_{130}$: Neither one of the words
$u_1=0\mu(1\mu(\mu(\mu({\bf x}))))=010\mu^4({\bf x})$,
$u_2=0\mu(1\mu(\mu({\bf x})))=010\mu^3({\bf x})$ contains an image of
the other. The proofs for both directions are essentially the same, so
we give only one of them. Let $u_1$ be \gfe; then the words $\mu^4({\bf
x})$, ${\bf x}$, and $\mu^3({\bf x})$ are \gfg as well, and $x$ starts
0. A simple inspection of short prefixes of $u_2$ shows that if this
word is not \gfe, then some $\beta$-power with $\beta\ge2$ is a prefix
of $\mu^3({\bf x})$. By Lemma~\ref{cor}, the word ${\bf x}$ also starts
with a $\beta$-power. The argument below will be repeated, with small
variations, for several identities.

\begin{itemize}
\item[$(*)$] Consider a prefix $yyy'$ of ${\bf x}$ which is the longest
prefix of ${\bf x}$ with period $|y|$. Then $|y'|<|y|/3$. By
Lemma~\ref{cor}, the longest prefix of the word $\mu^3({\bf x})$ having
period $8|y|$ is $\mu^3(yyy')$. If some word $z\mu^3(yyy')$ also has
period $8|y|$, then $z$ should be a suffix of a $\mu^3$-image of some
word. Since the word $010$ is not such a suffix, then $10\mu^3(yyy')$
is the longest possible $(8|y|)$-periodic word contained in $u_2$. Let
us estimate its exponent. Since $|z|$ and $|y'|$ are integers, we have
$$
8|y'|<8|y|/3\Longrightarrow 24|y'|<8|y|\Longrightarrow
24|y'|+6<8|y|\Longrightarrow 8|y'|+2<8|y|/3,
$$ 
whence $\exp(10\mu^3(yyy'))<7/3$. Since we have chosen an arbitrary
prefix $yyy'$ of ${\bf x}$, we conclude that the word $u_2$ is \gfe.
\end{itemize}

\smallskip

$F_{1301}=F_{1}$: The word $u_1=0\mu(1\mu(\mu(0\mu({\bf
x}))))=010\mu^3(0\mu({\bf x}))$ contains a $\mu^3$-image of
$u_2=0\mu({\bf x})$. Suppose $u_2$ is \gfe. It suffices to check that
the prefix $010$ of $u_1$ does not complete any prefix of $\mu^3(u_2)$
to a forbidden power. For short prefixes, this can be checked directly,
while long prefixes that can be completed in this way should have
exponents $\ge 2$. By Lemma~\ref{cor}, a prefix of period $p$ and
exponent $\beta\ge 2$ of the word $\mu^3(u_2)$ corresponds to the
prefix of the word $u_2$ having the exponent $\beta$ and the period
$p/8$. So, we repeat the argument $(*)$ replacing ${\bf x}$ with $u_2$
to obtain that $u_1$ is \gfe.

\smallskip

$F_{1302}=F_{2}$: We have $u_1=0\mu(1\mu(\mu(00\mu({\bf
x}))))=010\mu^3(00\mu({\bf x}))=010\mu^3(u_2)$. Suppose $u_2$ is \gfe.
By $(\star)$ and Lemma~\ref{cor}, among the prefixes of $\mu^3(u_2)$
there are only two squares, $\mu^3(00)$ and $\mu^3(001001)$, and no
words of bigger exponent. By direct inspection, $u_1$ is $(7/3)$-free.

\smallskip

$F_{2030}=F_{310}$: Neither one of the words
$u_1=00\mu(\mu(1\mu(\mu({\bf x}))))=001001\mu^4({\bf x})$ and
$$u_2=1\mu(0\mu(\mu({\bf x})))=101\mu^3({\bf x})$$ contains an image of
the other. If the word $u_1$ is assumed to be \gfe, then the proof
repeats the proof of the identity $F_{1300}=F_{130}$, up to renaming
all $0$'s to $1$'s and vice versa. Let $u_2$ be \gfe. The words
$\mu^4({\bf x})$ and ${\bf x}$ are also \gfe, and ${\bf x}$ begins with
$1$, assuring that there are no short forbidden powers in the beginning
of $u_1$.  Concerning long forbidden powers, we consider, similar to
$(*)$, a prefix $yyy'$ of ${\bf x}$ which is the longest prefix of
${\bf x}$ with period $|y|$. The longest possible $(16|y|)$-periodic
word contained in $u_1$ is $01001\mu^4(yyy')$, because $001001$ is not
a suffix of a $\mu^4$-image. As in $(*)$, we obtain $16|y'|+5<16|y|/3$,
implying $\exp(01001\mu^4(yyy'))<7/3$. Hence the word $u_1$ is \gfe.

\smallskip

$F_{2033}=F_{33}$: The word
$u_1=00\mu(\mu(1\mu(1\mu(x))))=00\mu^2(110\mu^2(x))$ contains a
$\mu^2$-image of $u_2=1\mu(1\mu(x))=110\mu^2(x)$.  Again, if the word
$u_2$ is \gfe, then so is $\mu^2(u_2)$, and it suffices to check that
the prefix $00$ of $u_1$ does not complete any prefix of $\mu^2(u_2)$
to a forbidden power. Similar to $(*)$, consider a prefix $yyy'$ of
$u_2$ which is the longest prefix of $u_2$ with period $|y|$. The
longest possible $(4|y|)$-periodic word contained in $u_1$ is
$0\mu^2(yyy')$, because $00$ is not a suffix of a $\mu^2$-image. As in
$(*)$, we see that $4|y'|+1<4|y|/3$, implying $\exp(0\mu^2(yyy'))<7/3$,
and conclude that the word $u_1$ is \gfe.

\smallskip

$F_{2034}=F_{4}$: The word
$u_1=00\mu(\mu(1\mu(11\mu(x))))=001001\mu^3(11\mu(x))$ contains a
$\mu^3$-image of $u_2=11\mu(x)$. Suppose $u_2$ is \gfe. Using $(\star)$
and Lemma~\ref{cor}, we conclude that among the prefixes of
$\mu^3(u_2)$ there are only two squares, $\mu^3(11)$ and
$\mu^3(110110)$, and no words of bigger exponent. By direct inspection,
$u_1$ is $(7/3)$-free.
\end{proof}

\begin{figure}[!htb]
\centerline{
\begin{picture}(110,92)(-5,-1)
\gasset{Nw=7,Nh=7,AHangle=15,AHLength=2.5,loopCW=n,loopdiam=7,ELdist=0.5,linewidth=0.1}
\node(e)(50,45){$\epsilon$}
\node(3)(32,55){\small$3$}
\node(1)(32,35){\small$1$}
\node(33)(16,65){\small$33$}
\node(11)(16,25){\small$11$}
\node(31)(0,80){\small$31$}
\node(13)(0,10){\small$13$}
\node(310)(50,75){\small$310$}
\node(130)(50,15){\small$130$}
\node(4)(68,55){\small$4$}
\node(2)(68,35){\small$2$}
\node(40)(85,53){\small$40$}
\node(20)(85,37){\small$20$}
\node(203)(100,65){\small$203$}
\node(401)(100,25){\small$401$}
\drawloop[loopangle=0](e){\scriptsize$0$}
\drawloop[loopangle=225](33){\scriptsize$3$}
\drawloop[loopangle=135](11){\scriptsize$1$}
\drawloop[loopangle=90](310){\scriptsize$0$}
\drawloop[loopangle=270](130){\scriptsize$0$}
\drawedge[curvedepth=2,ELside=r](e,1){\scriptsize$1$}
\drawedge[curvedepth=-2,ELside=r](e,3){\scriptsize$3$}
\drawedge[curvedepth=2](1,e){\scriptsize$0$}
\drawedge[curvedepth=-2](3,e){\scriptsize$0$}
\drawedge[curvedepth=-2](e,2){\scriptsize$2$}
\drawedge[curvedepth=2](e,4){\scriptsize$4$}
\drawedge[curvedepth=-2](1,2){\scriptsize$2$}
\drawedge[curvedepth=2](3,4){\scriptsize$4$}
\drawedge[ELside=r,ELpos=40](1,11){\scriptsize$1$}
\drawedge[curvedepth=4,ELpos=37](1,13){\scriptsize$3$}
\drawedge(11,3){\scriptsize$0$}
\drawedge[ELside=r](11,13){\scriptsize$3$}
\drawedge[ELpos=40](3,33){\scriptsize$3$}
\drawedge[curvedepth=-4,ELside=r,ELpos=40](3,31){\scriptsize$1$}
\drawedge[ELside=r](33,1){\scriptsize$0$}
\drawedge[ELside=r](33,31){\scriptsize$1$}
\drawedge[curvedepth=4,ELside=r](31,13){\scriptsize$3$}
\drawedge[curvedepth=4,ELside=r](13,31){\scriptsize$1$}
\drawedge[curvedepth=6](31,310){\scriptsize$0$}
\drawedge[ELside=r,ELpos=45](310,3){\scriptsize$3$}
\drawedge[ELpos=45](310,4){\scriptsize$4$}
\drawedge[curvedepth=30](33,4){\scriptsize$4$}
\drawedge[curvedepth=-6](13,130){\scriptsize$0$}
\drawedge[ELside=r,ELpos=40](130,1){\scriptsize$1$}
\drawedge[ELpos=40](130,2){\scriptsize$2$}
\drawedge[curvedepth=-30](11,2){\scriptsize$2$}
\drawedge[curvedepth=-6](4,31){\scriptsize$1$}
\drawedge[curvedepth=6,ELside=r](2,13){\scriptsize$3$}
\drawedge[curvedepth=3,ELside=r,ELpos=30](203,33){\scriptsize$3$}
\drawedge[curvedepth=-3,ELside=r,ELpos=30](401,11){\scriptsize$1$}
\drawedge(4,40){\scriptsize$0$}
\drawedge[ELside=r,ELpos=30](40,2){\scriptsize$2$}
\drawedge(40,401){\scriptsize$1$}
\drawedge[ELside=r,ELpos=40](401,2){\scriptsize$2$}
\drawedge[ELside=r](401,130){\scriptsize$0$}
\drawedge(2,20){\scriptsize$0$}
\drawedge[ELside=r,ELpos=30](20,4){\scriptsize$4$}
\drawedge(20,203){\scriptsize$3$}
\drawedge[ELpos=40](203,4){\scriptsize$4$}
\drawedge[ELside=r](203,310){\scriptsize$0$}
\end{picture} }
\caption{Automaton coding infinite binary \gfg words} \label{figure1}
\end{figure}
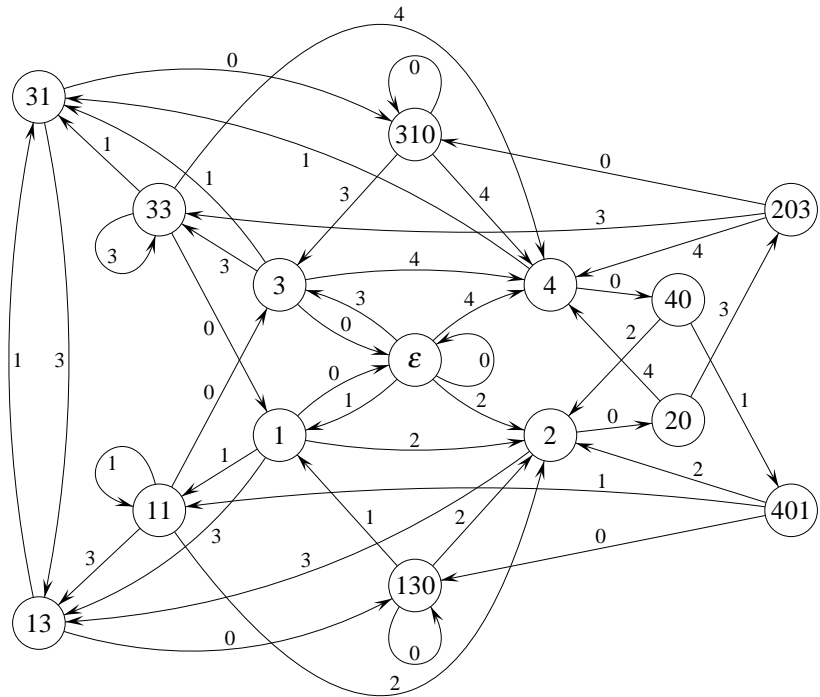

From Lemma~\ref{main2} and the results above, we get

\begin{theorem}
    Every infinite binary \gfg word $\bf x$ is encoded by an
infinite path, starting in $F_{\epsilon}$,
through the automaton in Figure~\ref{figure1}.

   Every infinite path through the automaton
not ending in $0^\omega$ codes a unique
infinite binary \gfg word $\bf x$.  If a path $\bf i$ ends in
$0^\omega$ and this suffix corresponds to a cycle on state $F_{\epsilon}$
then $\bf x$ is 
coded by either ${\bf i}; 1$ or
${\bf i}; 3$.  If a path $\bf i$ ends in $0^\omega$
and this suffix corresponds to a cycle on $F_{310}$,
then $\bf x$ is coded by ${\bf i}; 3$.  If a path $\bf i$ ends
in $0^\omega$ and this suffix corresponds to a cycle 
on $F_{130}$, then $\bf x$ is coded by ${\bf i}; 1$.
\label{main}
\end{theorem}

\begin{remark}
Blondel, Cassaigne, and Jungers \cite{Blondel&Cassaigne&Jungers:2009} 
obtained a similar result, and even more general ones,
for finite words.  The main 
advantage to our construction is its simplicity.
\end{remark}

\begin{corollary}
Each of the 15 sets $F_{\epsilon}$,
$F_1$, $F_2$, $F_3$, $F_4$,
$F_{11}$, $F_{33}$, $F_{13}$, $F_{31}$, $F_{20}$, $F_{40}$,
$F_{130}$, $F_{310}$, $F_{203}$, $F_{401}$ is uncountable.
\end{corollary}

\begin{proof}
It suffices to provide uncountably
many distinct paths from each state
to itself. By symmetry,
it suffices to prove this for all the states labeled $\epsilon$ or
below in Figure~\ref{figure1}.
These are as follows:
\begin{itemize}
\item $\epsilon$: $(0+10)^\omega$
\item $1$: $(01+001)^\omega$
\item $2$: $(0402+030402)^\omega$
\item $11$: $(0011+00011)^\omega$
\item $13$: $(013+0013)^\omega$
\item $20$: $(4020+34020)^\omega$
\item $401$: $(10401+203401)^\omega$
\item $130$: $(0+104010)^\omega$.   
\end{itemize}
\end{proof}

\begin{corollary}
For all words $w \in \lbrace 0,1,2,3,4 \rbrace^*$, either
$F_w$ is empty or uncountable.
\end{corollary}

\section{The lexicographically least $\frac{7}{3}$-power-free word}

\begin{theorem}
The lexicographically least infinite binary \gfg word  is
$0 0 1 0 0 1 \overline{\bf t}$.
\end{theorem}

\begin{proof}
By tracing through the possible paths through the 
automaton we easily find that $2030^\omega; 1$ is the code for the
lexicographically least sequence.
\end{proof}

\begin{remark}
This result does not seem to follow directly from
\cite{Allouche&Currie&Shallit:1998} as one referee suggested.
\end{remark}

\section{Automatic infinite binary $\frac{7}{3}$-power-free words}

As a consequence of Theorem~\ref{main}, we can give a complete description
of the infinite binary \gfg words that are $2$-automatic
\cite{Allouche&Shallit:2003}.  Recall that an infinite word $(a_n)_{n \geq 0}$
is $k$-automatic if there exists a deterministic finite automaton with
output that, on input $n$ expressed in base $k$, produces an output
associated with the state last visited that is equal to $a_n$.  Alternatively,
$(a_n)_{n \geq 0}$
is $k$-automatic if its $k$-kernel
$$\lbrace (a_{k^i n + j})_{n \geq 0} \ : \ i \geq 0 \text{ and } 
0 \leq j < k^i \rbrace$$
consists of finitely many distinct sequences.

\begin{theorem}
An infinite binary \gfg word is $2$-automatic if and only
if its code is both
specified by the DFA given above in Figure~\ref{figure1},
and is ultimately periodic.
\label{auto}
\end{theorem}

First, we need a lemma:

\begin{lemma}
An infinite binary word ${\bf x} = a_0 a_1 a_2 \cdots$
is $2$-automatic if and only if
$\mu({\bf x})$ is $2$-automatic.
\label{tm}
\end{lemma}

\begin{proof}
Proved in \cite{Shallit:2011}.  
\end{proof}

Now we can prove Theorem~\ref{auto}.

\begin{proof}
Suppose the code of $\bf x$ is ultimately periodic.  Then we can write
its code as $y z^\omega$ for some finite words $y$ and $z$.  
Since the class of $2$-automatic sequences is closed under appending
a finite prefix \cite[Corollary 6.8.5]{Allouche&Shallit:2003}, by
Lemma~\ref{tm}, it suffices to show that the word coded by $z^\omega$
is $2$-automatic.

The word $z^\omega$ codes a \gfg word ${\bf w}$ satisfying ${\bf w} = t
\varphi ({\bf w})$, where $t$ is a finite word and $\varphi = \mu^k$.
Hence, by iteration, we get that ${\bf w} = t \varphi(t)
\varphi^2(t) \cdots$.   It is now easy to see that the $2$-kernel of
$\bf w$ is contained in 
$$S := \lbrace u \mu^i(v) \mu^{i+k}(v) \mu^{i+2k}(v) \cdots \ : \ |u| \leq |t|
\text{ and } v \in \lbrace t, \overline{t} \rbrace \text{ and } 1 \leq i \leq k \rbrace,$$
which is a finite set.

On the other hand, suppose the code for $\bf x$ is not ultimately
periodic.  Then we show that the $2$-kernel is infinite.  Now it is
easy to see that if the code for $\bf x$ is $a{\bf y}$ for some letter
$a \in \lbrace 0, 1, 3 \rbrace$ then one of the sequences in the
$2$-kernel (obtained by taking either the odd- or even-indexed terms)
is either coded by $\bf y$ or its complement is coded by $\bf y$.  On
the other hand, if the code for $\bf x$ is $a{\bf y}$ with $a \in
\lbrace 2, 4 \rbrace$, then $\bf y$ begins with $0, 1, $ or $3$, say
${\bf y} = b {\bf z}$.  It follows that the subsequences
obtained by taking the terms congruent to $0, 1, 2,$ or $3$ (mod $4$)
is coded by $\bf z$, or its complement is coded by $\bf z$.  Since the
code for $\bf x$ is not ultimately periodic, there are infinitely many
distinct sequences in the orbit of the code for $\bf x$, under the
shift.  By the infinite pigeonhole principle, infinitely many
correspond to a sequence in the $2$-kernel, or its complement.  Hence
$\bf x$ is not $2$-automatic.  
\end{proof}

\section{Acknowledgments}

We are grateful to the referees for their helpful suggestions.

\bibliographystyle{eptcs}
\bibliography{fif}

\end{document}